\documentclass[runningheads,envcountsame]{llncs}
\usepackage{graphicx}
\usepackage{mathrsfs}
\usepackage{amsmath}
\usepackage{amssymb}
\usepackage{xcolor}

\usepackage{tikz-cd}
\usepackage[all,cmtip]{xy}
\usepackage{float}
\newcommand{\bbN}{\mathbb{N}}

\newcommand{\first}[1]{{f(#1)}}
\newcommand{\last}[1]{{\ell(#1)}}
\newcommand{\fact}[1]{{Fac(#1)}}
\newcommand{\suff}[1]{{Suf(#1)}}
\newcommand{\factsigma}[1]{{Fac_\sigma(#1)}}
\newcommand{\suffsigma}[1]{{Suf_\sigma(#1)}}

\newcommand{\pekna}{\mathcal{C}}
\newcommand{\wqo}{\textit{wqo}}
\newcommand{\lesigma}{\leq_{\sigma}}
\newcommand{\naomeg}{^{\omega}}

\newcommand{\gr}[1]{\mathscr{#1}}
\newcommand{\pv}[1]{\mathsf{#1}}

\newcommand{\eval}{\operatorname{eval}}

\mathchardef\mhyphen="2D

\makeatletter
\newcommand{\xRightarrow}[2][]{\ext@arrow 0359\Rightarrowfill@{#1}{#2}}
\makeatother

\usepackage[mathscr]{euscript}

\newenvironment{lemmaagain}[1]{\begin{trivlist}
\item[] \textbf{Lemma~\ref{#1}.}}{\end{trivlist}}
\newenvironment{propositionagain}[1]{\begin{trivlist}
\item[] \textbf{Proposition~\ref{#1}.}}{\end{trivlist}}
\newenvironment{theoremagain}[1]{\begin{trivlist}
\item[] \textbf{Theorem~\ref{#1}.}}{\end{trivlist}}
\newenvironment{coragain}[1]{\begin{trivlist}
\item[] \textbf{Corollary~\ref{#1}.}}{\end{trivlist}}
\newcommand{\dukazjinde}{\hfill\ensuremath{\spadesuit}}
\newcommand{\dukazjindeintro}{\ensuremath{\spadesuit}}

\newcommand{\uvoz}[1]{``#1"} 
\begin{document}
%

\title{Well Quasi-Orders Arising from Finite \\ Ordered Semigroups
\thanks{The research was supported by Grant 19-12790S of the Grant
Agency of the Czech Republic.}} 
\titlerunning{Well Quasi-Orders Arising from Finite Ordered Semigroups}
%
\author{Ond\v rej Kl\'ima
\and
Jonatan Kolegar
}  
\authorrunning{O. Kl\'ima and J. Kolegar}
%
\institute{Department of Mathematics and Statistics, Masaryk University \\ 
Kotl\'a\v rsk\'a 2, 611 37 Brno, Czech Republic \\
\email{klima@math.muni.cz, kolegar@math.muni.cz}
}
\maketitle              
\begin{abstract}

In 1985, Bucher, Ehrenfeucht and Haussler studied derivation relations associated with a given set of context-free rules. 
Their research motivated a question regarding homomorphisms from the semigroup of all words onto a finite ordered semigroup. 
The question is which of these homomorphisms induce a well quasi-order on the set of all words. We show that this problem is decidable and the answer does not depend on the homomorphism, but it is a property of the ordered semigroup.

\end{abstract}


\section{Introduction}
\label{s:introduction}

The notion of well quasi-order (\wqo) is a well-established tool in mathematics and in many areas of theoretical computer science that was rediscovered by many authors 
(see~\cite{kruskal} by Kruskal). A comprehensive overview of the applications of the notion in theory of formal languages and combinatorics on words can be found in the book~\cite{delucavarri} by de Luca and Varricchio or in the survey 
paper~\cite{alessvarr-dlt08} by D'Alessandro and Varricchio.
Since our contribution belongs to formal language theory, 
we recall the central notion of \wqo\ directly for the set of all words 
over a finite alphabet $A$. A quasi-order $\le$ on a set $A^*$ is  \wqo\ if it has 
no infinite antichains and no decreasing infinite chain (the latter property is often called \emph{well-foundness}). 
There are several equivalent conditions of the notion (see, \emph{e.g.},
~\cite[Theorem 6.1.1]{delucavarri}); among them we recall the following: for every infinite
sequence of words $w_1,w_2, \dots  $ there exist integers $0<i<j$ such that $w_i \le w_j$.   
We point out that the important property which makes the notion of \wqo\ a useful tool is that every upper closed subset of $A^*$ with respect to a \wqo\ $\le$ is a regular language
(see~\cite[Theorem 6.3.1]{delucavarri}). 

The first example of \wqo\ in the area of formal languages was 
given by Higman~\cite{higman}. We mention the simplest 
consequence of the general statement, namely the result that the \emph{embedding} 
relation $\unlhd$ on $A^*$ is \wqo.
The embedding relation $\unlhd$ is often called subword ordering, because
a word $u$ \emph{embeds} in a word $v$ if $u$ is a scattered subword of $v$, 
\emph{i.e.}, $u\unlhd v$ if there are factorizations of the same length 
$u=a_1\dots a_k$ and $v=v_1\dots v_k$ such that, for all $i\in \{1,\dots, k\}$, we have 
$a_i\in A$, $v_i\in A^+$, and $a_i$ appears in $v_i$.     

The considered notion of embedding relation can be modified by 
requiring different conditions on the factorizations. 
For example, if the alphabet $A$ is quasi-ordered by $\preceq$, then we may replace 
the condition that $a_i$ appears in $v_i$ by the condition that there is 
$j$ such that $a_i\preceq a_j$ and $a_j$ appears in $v_i$. 
In this way we obtain a quasi-order which is a well known (and more general) instance of Higman result. 
Another variant is the following \emph{gap embedding} considered by 
Sch\"{u}tte and S. G. Simpson in~\cite{ss85} for an alphabet equipped by linear
order $\sqsubseteq$: the defining 
condition that $a_i$ appears in $v_i$  is replaced by the condition that the letter 
$a_i$ is the last letter in $v_i$ and it is a minimal 
letter in $v_i$ with respect to $\sqsubseteq$. Notice that yet another modification
of gap embedding is the \emph{priority embedding} 
in~\cite{hss04} by  Haase, Schmitz, and Schnoebelen.

Our paper concentrates on an application of \wqo s\ 
motivated by the work of Bucher, Ehrenfeucht and Haussler \cite{bucher},
which leads to a purely algebraic question in the realm of ordered semigroups. 
Notice that the topic is nicely overviewed in the recent survey 
paper~\cite{pin-lata-2020} by Pin.

Before we introduce the primary question, we briefly recall the role of ordered semigroups in the algebraic theory of regular languages. At first, when we talk about an ordered semigroup $(S,\cdot,\le)$, we assume that the partial order $\le$ is stable, \textit{i.e.}, compatible with the multiplication $\cdot$
in the sense that, for arbitrary $x,y,s\in S$, the inequality $x\le y$ implies both 
$s\cdot x\le  s\cdot y$ and $x\cdot s \le y\cdot s$. The finite ordered semigroups 
are used to recognize regular languages similarly to unordered semigroups -- see,
{\it e.g.}, the fundamental survey on the algebraic theory of regular languages~\cite{pin-handbook} by Pin.
The modification is natural as the syntactic semigroup of a regular language is implicitly ordered in the following way. In the syntactic congruence of the regular language, 
words are related if they have the same set of contexts putting the words into the language. 
Then one may also compare these sets of contexts by the inclusion relation; this comparison
gives the syntactic quasi-order and consequently the partial order on the syntactic semigroup of the considered regular language.
Let us note that in the literature the syntactic quasi-order is not always defined in this way, but the dual quasi-order is considered instead, \textit{e.g.}, in \cite{pin-handbook}.

The starting point in the study of well quasi-orders in \cite{bucher} was a research by Ehrenfeucht, Haussler, and Rozenberg~\cite{ehrenfeuchtHR83} concerning  certain rewriting systems preserving regularity, 
where a well quasi-order plays a role of a sufficient condition guaranteeing the required property of the rewriting system. 
Particular attention in that research is paid to rewriting systems $R$ with rules of the form $a\rightarrow u$ with $a$ being a letter and $u$ being a word.
For such a rewriting system, in~\cite{bucher}, there are stated equivalent conditions to the fact that the derivation relation $\xRightarrow{*}_R$ is a well quasi-order. For example, one of the equivalent conditions is that 
the set $L=\{aua \mid a\in A, u\in A^*, a \xRightarrow{*}_R aua \}$ is unavoidable (in the sense that every infinite word over the alphabet $A$ contains a finite factor from the language $L$).
Unfortunately, they did not give algorithms to test the conditions.
They also showed that a derivation relation that is a \wqo\ comes from
a rewriting system induced by a semigroup homomorphism $\sigma : A^+ \rightarrow S$ onto a finite ordered semigroup $(S,\cdot,\le)$ by the following formula:
$$R_\sigma=\{ a\rightarrow u \mid a\in A, u\in A^+, \sigma (a)\le \sigma (u)\}.$$ 
Finally, the open question is to characterize 
homomorphisms $\sigma$ such that $\xRightarrow{*}_{R_{\sigma}}$ is a well quasi-order. 
Another research goal is a characterization of those finite ordered semigroups $S$, 
such that the relation $\xRightarrow{*}_{R_{\sigma}}$ is a \wqo\ for every 
alphabet $A$ and a homomorphism $\sigma : A^+ \rightarrow S$. For the purpose of this paper we call these ordered semigroups \emph{congenial}.

Let us note that the examples of \wqo\ mentioned earlier also fit to the introduced scheme of relations arising from 
a homomorphism onto a finite ordered semigroup. 
Indeed, for an alphabet $A$ 
we may consider an ordered semigroup $(P(A),\cup,\subseteq)$ consisting of non-empty subsets of $A$ equipped with the operation of union, and ordered by the inclusion relation. Taking then the homomorphism $\sigma : A^+ \rightarrow P(A)$, where $\sigma(a)=\{a\}$,
the above considered relation  $\xRightarrow{*}_{R_{\sigma}}$ coincides with 
the embedding relation $\unlhd$. 
Similarly, we can consider an appropriate homomorphism onto an ordered semigroup in the other mentioned examples.\footnote{For a variant of Higman's Lemma where the alphabet $A$ is equipped with the quasi-order $\preceq$, we take for the ordered semigroup $S$ the subsemigroup of $P(A)$ consisting of all downward closed subsets of $A$
with respect to the considered quasi-order $\preceq$.
For the gap embedding, one may construct the semigroup $A\times A$ (ordered by equality), 
where the multiplication is given by $(a,b)\cdot (c,d)=(min(a,c),d)$, where $min$ is taken with respect to $\sqsubseteq$ (and the homomorphism $\sigma\colon A\to A\times A$ is the diagonal mapping).
}

Up to our knowledge, and also according to the survey paper~\cite{pin-lata-2020},
there is just one significant contribution to the mentioned open questions.
Namely, in the paper~\cite{kunc} by Kunc, the questions are solved for the semigroups ordered by the equality relation.
It is stated in~\cite{kunc} (implicitly contained in the proof of Theorem 10) that the property 
depends only on the semigroup $S$, not on the actual homomorphism $\sigma$.
The congenial semigroups ordered by the equality relation are characterized as finite chains of finite simple semigroups. One of the equivalent characterizations of this transparent structural property is the following condition which can be checked in polynomial time: 
for every $s,t\in S$, we have $(s\cdot t)^\omega \cdot s=s$ or $t\cdot (s\cdot t)^\omega =t$, where $(s\cdot t)\naomeg$ is the power of $s\cdot t$ which is idempotent.

Our research aims to give an analogous characterization in the general case; however, we do not fulfill that program yet, and our contribution brings tentative results. 
As the main result, we show that the problem of whether a homomorphism induces a well quasi-order is decidable. 
The proof is almost a straightforward application of the mentioned characterization from~\cite{bucher}. 
Furthermore, we give some necessary conditions related to the studied property, and, on the other hand,
we establish other conditions which ensure the property.
Next, we show that the mentioned side result from~\cite{kunc} holds in the full generality: 
the property is indeed a property of an ordered semigroup and does not depend on the homomorphism.

The proofs of results marked by the symbol $\dukazjindeintro$ are available in Appendix.


\section{Preliminaries}
\label{s:preliminaries}
We briefly recall basic notions and fix notation used in the paper. We start with that from 
semigroup theory.
When we talk about ordered semigroup or semigroup, 
we write simply $S$ instead of formal notation $(S,\cdot,\le)$ and
 $(S,\cdot)$.
Throughout the paper we work with finite semigroups with the exception of the free 
monoid $A^*$ (the free semigroup $A^+$) formed by (non-empty) words over an alphabet $A$.
We use the symbol $\varepsilon$ for the empty word.
An element $e$ in a semigroup $S$ is called \emph{idempotent} if  $e\cdot e=e$.
For all $s\in S$, the set $\{s^n\,|\,n\in \mathbb N\}$ 
contains exactly one idempotent, which is denoted $s\naomeg$.
We put $s^{\omega+1}=s^\omega\cdot s$ which equals (by 
definition) to $s\cdot s^\omega$.
By $S^1$ we mean the monoid $S\cup \{1\}$ with a 
new neutral element $1$ added
when  $S$ is not a monoid and $S^1=S$ otherwise.
We denote $\eval_S\colon S^+\to S$ the evaluation homomorphism from the free semigroup over $S$ defined by the rule $\eval_S(s)=s$ for all $s\in S$. Here, an element $w\in S^+$ is a word $w=s_1s_2\dots s_k$, where $s_i\in S$ for $i\in\{1,\dots,k\}$, and for such $w$ we have $\eval_S(w)=\eval_S(s_1s_2\dots s_k)=s_1\cdot s_2\cdots s_k$.

Furthermore, we use the Green relations, a basic notion in the theory of semigroups (see \cite{howie} by Howie). 
For the reader's convenience we recall that an \textit{ideal} of a semigroup $S$ is a non-empty subset $I\subseteq S$ such that for all $t\in I$ and $s\in S$ we have $t\cdot s\in I$ and $s\cdot t\in I$. 
The ideal generated by an element $s\in S$ is equal to $S^1sS^1=\{x\cdot s\cdot y\,|\, x,y\in S^1\}$. 
Then the Green relation $\gr{J}$ is defined by the rule  
$s\,\gr{J}\,t \Longleftrightarrow S^1sS^1= S^1tS^1$.
We say that a semigroup is \textit{simple} if it has no proper ideal, \textit{i.e.}, if all elements of the semigroup are $\gr{J}$-equivalent. 

The following lemma is well known (see, \textit{e.g.}, \cite[Theorem 1.11]{pin-varieties} by Pin).
\begin{lemma}\label{l:folk}
Let $S$ be a finite semigroup. There exists $n\in\bbN$ such that 
for every sequence of elements $s_1,\dots,s_n\in S$ there are indices $i,j\in\{1,\dots, n\}, i\leq j$ for which the product $s_i\cdot s_{i+1}\cdots s_j$ is an idempotent.
\end{lemma}

For words $w,x,y,z\in A^*$ such that $w=xyz$, we say that $x$ is a \emph{prefix},
$y$ is a \emph{factor} and $z$ is a \emph{suffix}. 
We say that the prefix $x$ of $w$ is proper if $|x|<|w|$, where $|w|$ is the length of the word $w$.
A language $L\subseteq A^*$ is unavoidable if an arbitrary infinite word $u$ over $A$ has a factor $v$ in $L$. We denote the set of all infinite words over $A$ by $A^\infty$.

A \textit{quasi-order} $\leq$ on a set $X$ is a reflexive and transitive binary relation.
It is called a \textit{well quasi-order (wqo)} if for an infinite sequence $(x_n)_{n\in\bbN}$ of elements of $X$ there exist indices $m,n\in\bbN$ such that $m<n$ and $x_m\leq x_n$. 
Many equivalent defining conditions are known (see, {\it e.g.},~\cite[Theorem 6.1.1]{delucavarri}). 

The next definition is partially motivated by results from \cite{bucher}. We prefer to follow the formalism and notation used in \cite{kunc}.

\begin{definition}
Let $\sigma\colon A^+\to S$ be a homomorphism onto a finite ordered semigroup. We denote $\lesigma$ a quasi-order on $A^*$ defined by setting $u \lesigma v$ if and only if there exist factorizations $u=a_1\dots a_n$ and $v=v_1\dots v_n$, such that for all
$i\in \{1,\dots, n\}$ we have $a_i\in A, v_i\in A^+$ and $\sigma(a_i)\leq\sigma(v_i)$. 
\end{definition}

We refer to the list of 
inequalities $\sigma(a_j)\leq\sigma(v_j)$ for $j\in\{1,\dots,n\}$ as to the \textit{proof of} $u\lesigma v$, and we can use the proof to form other inequalities. We say that
$a_i\dots a_j \lesigma v_i\dots v_j$ 
is the \textit{consequence of the proof given by the factor} $a_i\dots a_j$.
Note that $u\lesigma v$ implies $\sigma(u)\le \sigma(v)$ and either $u=v=\varepsilon$ or $u,v\in A^+$. 
Finally, it is clear that $\lesigma$ is a stable quasi-order on $A^*$.

The main result from~\cite{bucher} has the following natural interpretation. To see it, 
one needs other results~\cite[Section 3]{bucher} concerning rewriting systems of 
the form mentioned in the introduction (called \textit{OS scheme} in~\cite{bucher}). 
If the relation $\xRightarrow{*}_R$ with rules of the form $a\rightarrow u, a\in A, u\in A^+$ is a \wqo, then there is a homomorphism $\sigma : A^+ \rightarrow S$ onto a finite ordered semigroup such that the relations $\xRightarrow{*}_R$ and $\le_\sigma$ coincide. 
And vice versa, if $\le_\sigma$ is a \wqo, then there is a system $R$ with the property 
$\xRightarrow{*}_R\ =\ \le_\sigma$.

\begin{proposition}[\cite{bucher}]\label{p:bucher}
Let $\sigma\colon A^+\to S$ be a homomorphism onto a finite ordered semigroup. Then the following conditions are equivalent:
\begin{enumerate}
\item\label{lesigma} The relation $\lesigma$ is a well quasi-order on $A^*$.
\item\label{twosides} The language $L_{\sigma} = \{awa\,|\,a\in A,w\in A^*, a\lesigma awa\}$ is unavoidable over $A$.
\item\label{oneside} The language $\{aw\,|\,a\in A,w\in A^*, a\lesigma aw\}$ is unavoidable over $A$.
\end{enumerate}
\end{proposition}

We say that an ordered semigroup $S$ is \textit{congenial} if for every homomorphism $\sigma\colon A^+\to S$ the corresponding relation $\lesigma$ is a well quasi-order. The class of congenial semigroups is denoted $\pekna$.

We finish this section with a basic observation that 
it is enough to consider the case of the homomorphism 
$\eval_S$ when a congeniality of $S$ is tested.
We establish the following auxiliary lemma with the proof essentially 
same as to the unordered case (see~\cite[Theorem 10,\,(iii)$\Longrightarrow$(i)]{kunc}).

\begin{lemma}
\label{l:vetsi-abeceda}
Let $\sigma \colon A^+ \rightarrow S$ be a homomorphism to an ordered semigroup $S$ such that 
$\le_\sigma$ is a \wqo. Let $B$ be an alphabet, 
$\alpha : B^+ \rightarrow A^+$ be a homomorphism of free semigroups such that 
$\alpha (B)\subseteq A$, and $\varphi=\sigma\circ \alpha$. 
Then the quasi-order $\le_\varphi$ is a \wqo.\dukazjinde
\end{lemma}

The following statement is a direct consequence of Lemma~\ref{l:vetsi-abeceda}, where we take $A=S$ and $\sigma=\eval_S$.

\begin{lemma}
\label{l:evaluace}
A semigroup $S$ is congenial if and only if
$\leq_{\eval_S}$ is a \wqo.\dukazjinde
\end{lemma}

\section{What Makes an Ordered Semigroup Congenial}
\label{s:conditions}

Due to \cite[Lemma 2]{kunc}, we know that a semigroup $S$ ordered by equality is congenial if and only if for  every $s,t\in S$ either $s=(s\cdot t)\naomeg \cdot s \text{ or } t= t\cdot (s\cdot t)\naomeg.$
The natural generalization of this condition for an ordered semigroup $S$ is
\begin{equation}
\label{eqn:(1)}
\forall s,t\in S\colon s\leq (s\cdot t)\naomeg\cdot  s \text{ or } t\leq t\cdot (s\cdot t)\naomeg.
\end{equation}
We show that this condition is indeed necessary. 
Note that a semigroup satisfying~(\ref{eqn:(1)}) also satisfies $s\leq s^{\omega +1}$ as it is just the condition (\ref{eqn:(1)}) with $s=t$. 

\begin{proposition}\label{p:(1)}
Every congenial semigroup satisfies the condition (\ref{eqn:(1)}). \dukazjinde
\end{proposition}

Note that for the condition $s\leq s^{\omega+1}$ leaving the realm of 
ordered semigroups we get the equality $s=s^{\omega+1}$ defining 
%
the widely studied class of finite completely regular semigroups.

The following example shows that 
the condition~(\ref{eqn:(1)})
is not a sufficient condition. It indicates that the ordered situation is more complicated.  

\begin{example}\label{ex:potvurka}
Denote $\mathcal{F}_{LRB}(3)$ the free left-regular band (\textit{i.e.}, semigroup satisfying the identities $xyx=xy$ and $x^2=x$) over three generators $a,b,c$. 
The semigroup has 15 elements represented by words listed in Fig.~\ref{fig:potvurka}, where  
the order is depicted. 
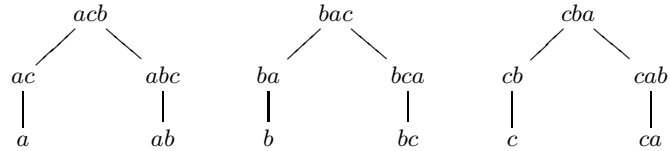
\begin{figure}[ht]
    \centering
\[
\xymatrix@C=3mm@R=4mm{
& acb\ar@{-}[dl] \ar@{-}[dr] &&&& bac \ar@{-}[dl] \ar@{-}[dr] &&&& cba\ar@{-}[dl] \ar@{-}[dr] \\
ac \ar@{-}[d]&& abc \ar@{-}[d] && ba \ar@{-}[d] && bca\ar@{-}[d] && cb\ar@{-}[d] && cab\ar@{-}[d] \\
a && ab && b && bc && c && ca
}
\]
\caption{The order $\leq$ of $\mathcal{F}_{LRB}(3)$.}
\label{fig:potvurka}
\end{figure}
For the product of a pair of elements, 
we simply concatenate the words and then omit the second occurrence of each letter 
if it occurs.
It is a routine to check that $\mathcal{F}_{LRB}(3)$  satisfies the condition (\ref{eqn:(1)}). 

Now we take $\sigma\colon \{a,b,c\}^+\to \mathcal{F}_{LRB}(3)$ where $\sigma(a)=a, \sigma(b)=b, \sigma(c)=c$. For the language $L_\sigma$ given by Proposition~\ref{p:bucher}, we see that $u\in L_\sigma$ if and only if $\sigma(u)\in \{a,ac,acb,b,ba,bac,c,cb,cba\}$.
The periodical infinite word generated by $abc$, that is $(abc)^{\infty}=abcabc\dots$, has no factor in the language $L_{\sigma}$ since $\sigma((abc)^n a)=abc$ (for $n\in\mathbb{N}$) and similarly for factors starting and ending with $b$, resp. $c$. This means that the language $L_{\sigma}$ is avoidable and $\mathcal{F}_{LRB}(3) \notin \pekna$. Therefore, the condition (\ref{eqn:(1)}) is not the characterization of the class of congenial semigroups. 
\end{example}

We also add an example of ordered semigroup which is not completely regular.

\begin{example}\label{ex:brandt}
We consider two ordered versions $B_2^+$ and $B_2^-$ of the Brandt semigroup $B_2$.
The semigroup is generated by two elements $a$ and $b$ satisfying $a^2=0$, $b^2=0$, $aba=a$, and $bab=b$. The semigroup has five elements $a$, $b$, $ab$, $ba$, and~$0$, where $ab,ba,0$ are idempotents.
The orders are given in Fig. \ref{fig:brandt}.
\begin{figure}[th]
    \centering
\[
\xymatrix@C=2mm@R=2mm{
&a \ar@{-}[ddrrr] && ab \ar@{-}[ddr] && ba \ar@{-}[ddl] && b \ar@{-}[ddlll] &&&&&&& 0\\
B_2^+: &&&&&&&&&& B_2^-: &&&&\\
&&&& 0 &&&&&&& b \ar@{-}[uurrr] && ba \ar@{-}[uur] && ab \ar@{-}[uul] && a \ar@{-}[uulll]
}
\]
    \caption{Orders of $B_2^+$ and $B_2^-$.}
    \label{fig:brandt}
\end{figure}
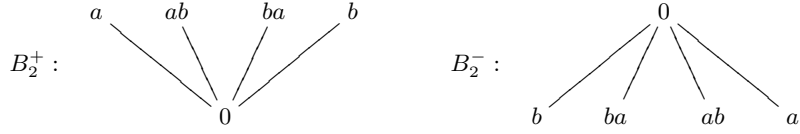
In both cases we consider a homomorphism $\sigma\colon \{a,b\}^+\to B_2,$ where $\sigma(a)=a,\, \sigma(b)=b$.

Firstly, we deal with $B_2^+$. Taking the sequence of words $(a^i)_{i\in\mathbb N}$, we get an infinite antichain with respect to
 $\lesigma$ showing that $\lesigma$ is not a $\wqo$ and thus $B_2^+\notin\pekna$.
For the ordered semigroup $B_2^-$, we see that
$a\leq a^2 =0, b\leq b^2=0$, and $a=aba$, and so
$a^2, b^2, aba \in L_{\sigma}$ for $L_\sigma$ from Proposition \ref{p:bucher}. The language $\{a^2,aba, b^2\}$ is unavoidable, which implies $B_2^-\in\pekna$.
\end{example}

Motivated by the previous examples and basic observations, we show 
the first sufficient condition ensuring the congeniality.

\begin{proposition}\label{prop:silnejis-nez-1} 
Let $S$ be 
a finite ordered semigroup satisfying the inequality $ x \le x \cdot (y\cdot x )^\omega$.
Then $S$ is congenial.
\end{proposition}

\begin{proof} Let $\sigma : A^+ \rightarrow S$ be an arbitrary homomorphism.
We show that the language $L_{\sigma}=\{awa\,|\,a\in A,w\in A^*, \sigma(a)\leq \sigma(awa)\}$ from Proposition \ref{p:bucher} is unavoidable. Let $v$ be an infinite word over the alphabet $A$. 
Since the alphabet is finite, some letter $a\in A$ has infinitely many occurrences in $v$. We consider the factorization $v=w_0aw_1aw_2aw_3 \dots$, where the words $w_i$ do not contain the letter $a$. We take the sequence $s_1=\sigma(w_1a), s_2=\sigma(w_2a), \dots$ and use Lemma \ref{l:folk} to show that there exist indices $i,j$ such that $\sigma(w_ia\dots w_ja)$ is an idempotent. If we denote $x=\sigma(a)$, $y=\sigma (w_ia\dots w_j)$, then we get $y\cdot x= (y\cdot x)^\omega$ and $x\le x \cdot (y\cdot x)^\omega = x\cdot y\cdot x$. 
Therefore $a\lesigma aw_ia\dots w_j a$ and the infinite word $v$ has a factor in~$L_{\sigma}$. \qed 
\end{proof}

Natural examples of ordered semigroups satisfying the assumption in the previous 
statement can be found in Appendix.


\section{Effective Characterization of the Class $\pekna$}
\label{s:automaton}

In order to check the condition in Proposition~\ref{p:bucher}, we introduce some technical notation. Let $A$ be an alphabet, and $w\in A^+$ be a word. Then we write 
$\first w$ for the first letter in $w$, {\it i.e.}, 
the letter $a$ such that $w\in aA^*$. Dually, $\last w$ means the last letter in the word $w$. 
Moreover, we denote the set of all factors and suffixes of a given word in a usual way with an exception that we do not consider letters as factors and suffixes here:   
$$\fact w=\{u\in A^+\setminus A \mid \exists p,q\in A^*, w=puq\}, \quad \text{and}$$
$$\suff w=\{u\in A^+\setminus A \mid \exists p\in A^*, w=pu\}.$$

Now, let $\sigma : A^+ \rightarrow S$ be a homomorphism onto a finite ordered semigroup.
We introduce the main technical notation: for $w\in A^+$ we put
$$\factsigma w=\{ \left( \sigma(u),\first u, \last u\right) \in S\times A\times A \mid u\in \fact w \}, 
\quad \text{and}$$
$$\suffsigma w=\{(\sigma(u),\first u, \last u) \in S\times A\times A \mid u\in \suff w \}.$$
Notice that $\factsigma w =\suffsigma w =\fact w=\suff w=\emptyset$ whenever $w\in A$. 
Furthermore, $\suffsigma w \subseteq S\times A\times \{ \last w\}$ for every word $w$, 
that is a useful property motivating the following definition. A non-empty subset 
$M$ of the set $S\times A\times A $ is called \emph{coherent} if there is a letter 
$a\in A$ such that $M\subseteq S\times A\times \{a\}$; if such a letter exists, we denote it by $\last M$.

Clearly, $w\in A^+$ does not avoid $L_\sigma=\{ava \mid a\in A, v\in A^*, a\le_\sigma ava\}$ if and only if there exist $w',v,w''\in A^*$ such that $w=w'avaw''$ and $\sigma(a)\leq \sigma(ava)$.
The latter condition is equivalent to $(\sigma(ava),a,a)\in \factsigma w$ with $\sigma(a)\leq \sigma(ava)$.
In other words, a word $w\in A^+$ avoids the set 
$L_{\sigma}$
if and only if $\factsigma w$ is disjoint with
the set 
$F=\{(s,a,a) \in S\times A\times A \mid \sigma(a) \le s\}.$ 
Now, we are ready to formulate a direct consequence of Proposition~\ref{p:bucher}. 

\begin{lemma}
Let $\sigma : A^+ \rightarrow S$ be a homomorphism onto a finite ordered semigroup
$S$. Then the relation $\le_\sigma$ is a well quasi-order if and only if the set 
$\{ w\in A^+\setminus A \mid \factsigma w \cap F =\emptyset \}$
is finite. 
\qed
\end{lemma}

To test whether the considered set is finite, we use that 
every $\factsigma w$ is a subset of $S\times A\times A$, and therefore 
there are only finitely many of them. 
In fact, we compute all possible $\suffsigma w$ disjoint with $F$
instead of computing all $\factsigma w$.
It is enough as $\factsigma w$ is a union of all $\suffsigma u$ 
where $u$ is a prefix of $w$. 
Naturally, we compute sets $\suffsigma u$
recursively, since $\suffsigma{wa}$ can be determined by 
$\suffsigma{w}$ in the following way.  
Informally speaking, we add $a$ at the end  of all elements of $\suffsigma{w}$ 
and evaluate the suffix $\last{\suffsigma w}a$
of $wa$ of length two. 
Therefore, we see the sets $\suffsigma{w}$ as states of the following finite 
deterministic incomplete automaton $\mathcal{A}_{\sigma}$
over the alphabet $A$. Notice that the automaton does not have final states.
 
We put $\mathcal{A}_{\sigma}=(Q,A,\delta,\iota)$ where 
$Q=\{\iota\}\uplus \bar{A}\uplus\mathcal{P}$, 
$\bar{A}=\{\bar{a}\,|\,a\in A\}$, and $\mathcal{P}=\{M\subseteq S\times A\times A \,|\, M\neq \emptyset, M\cap F=\emptyset, M\ \text{coherent} \}$. For a given set $M\in\mathcal{P}$ and a letter $a\in A$ we define
$$ M\ast a = \{(s\cdot\sigma(a),b,a)\,|\,(s,b,c)\in M\}\cup \{(\sigma(\ell(M)a),\ell(M),a)\}.$$
Similarly, for $\bar{b}\in \bar{A}$ we put $\bar{b}\ast a= \{(\sigma(ba),b,a))\}$. 
Furthermore, we define the partial transition function $\delta : Q\times A \rightarrow Q$ 
by $\delta(\iota,a)=\bar{a}$ for the initial state $\iota$, and for $q\in Q\setminus\{\iota\}$ we put $\delta(q,a)=q\ast a$ if $q\ast a\in \mathcal P$.
Note that the condition $q\ast a\in \mathcal{P}$ is equivalent to $q\ast a\cap F=\emptyset$ since $q\ast a$ is always non-empty and coherent. In particular, 
we have $\ell(q\ast a)=a$.
As usual, the partial function $\delta$ can be extended to the partial function $\delta : Q\times A^+ \rightarrow Q$, which is denoted by $\delta$ too. 

The following lemma summarises the properties of the previous constructions, with an obvious proof by an induction with respect to the length of words.

\begin{lemma}
\label{l:automat}
Let $\sigma : A^+ \rightarrow S$ be a homomorphism onto a finite ordered semigroup
$S$, and $\mathcal A_\sigma$ be the automaton defined as above. 
For every word $w\in A^+\setminus A$, the state $\delta(\iota,w)$ is 
defined in $\mathcal A_\sigma$ if and only if $\factsigma w \cap F =\emptyset$. 
Moreover, if  $\delta(\iota,w)$ is defined, 
then $\delta(\iota,w)=\suffsigma w$. 
\qed
\end{lemma}

Now, we are ready to state the main result. The proof is straightforward consequence of Propositon~\ref{p:bucher} and the constructions and lemmas in this section. 

\begin{theorem}
\label{th:OKautomat}
Let $\sigma : A^+ \rightarrow S$ be a homomorphism onto a finite ordered semigroup
$S$.  Then $\le_\sigma$ is a \wqo\ if and only if the automaton 
$\mathcal{A}_{\sigma}$ does not contain 
an infinite path starting in the initial state $\iota$. \dukazjinde
\end{theorem}

The reader may look into Appendix for an illustrative example of the construction of $\mathcal{A}_\sigma$.
The purpose of Theorem~\ref{th:OKautomat} is the following statement.

\begin{corollary} 
\label{cor:decidability}
Let $\sigma : A^+ \rightarrow S$ be a homomorphism onto a finite ordered semigroup
$S$.  Then it is decidable whether $\le_\sigma$ is a well quasi-order. 
\qed
\end{corollary}

Recall that all states in $\mathcal P$ are coherent subsets of $S\times A\times A$.
Since the automaton $\mathcal{A}_\sigma$ is finite, the existence of an infinite path starting in the initial state $\iota$ is equivalent to the existence of a loop reachable from $\iota$.
If we assume that there is a loop labeled by $u$ and reachable by $v$, then we have 
that $vu^\infty=vuuu\dots$ avoids $L_{\sigma}$. Hence the periodical infinite word $u^\infty$ avoids $L_{\sigma}$ too. 

\begin{corollary} 
\label{cor:periodic-word}
Let $\sigma : A^+ \rightarrow S$ be a homomorphism onto a finite ordered semigroup
$S$.  Then there is an infinite word avoiding $L_\sigma$ if and only if there is a periodic infinite word $u^\infty$ with that property. 
\qed
\end{corollary}

The number of states of the automaton $\mathcal A_\sigma$ is bounded by 
$|A|\times 2^{|S|\times |A|}+|A|+1$, that gives the obvious exponential bound for the time complexity of the algorithm based on Theorem~\ref{th:OKautomat}.

One may modify the construction of $\mathcal A_\sigma$ if the condition \eqref{oneside} from Proposition~\ref{p:bucher} replaces the condition \eqref{twosides}.
This means that Corollary~\ref{cor:periodic-word} holds if we take the set $\{aw \,|\, a\in A, w\in A^+, \sigma(a)\le \sigma(aw)\}$  instead of the set $L_\sigma$.


\section{Other Necessary and Sufficient Conditions}
\label{sec:kombinatorika}

The motivation for this section is to examine whether the condition that $\le_\sigma$ is a \wqo\
depends on the homomorphism $\sigma$ or it is just a property of the semigroup. Therefore, we try to prove necessary conditions from Section~\ref{s:conditions} under the assumption that $\le_\sigma$ is a \wqo.

\begin{proposition}
\label{p:quasiCR}
Let $\sigma : A^+ \rightarrow S$ be a homomorphism onto a finite ordered semigroup
$S$ such that $\le_\sigma$ is a \wqo. Then for every $u\in A^+$ there exists an integer $p>1$ such that $u\le_\sigma u^p$.  
\end{proposition}
\begin{proof}
We show the statement by induction with respect to the length of the word $u$. 
For every $a\in A$ the definition of \wqo\ implies that $a^k\le_\sigma a^\ell$ for some integers $k<\ell$, and by the definition of $\le_\sigma$ we have $a\le_\sigma a^p$ for some $p>1$.

Assume that the statement is true for all words shorter
than a given word $u\in A^+\setminus A$. 
Similarly to the initial step, we have 
$u^k\le_\sigma u^\ell$ for some $k<\ell$. 
We consider the consequences of the proof of $u^k\le_\sigma u^\ell$ given by factors
$u$ of $u^k$. The first non-trivial inequality among these consequences 
in the order from left to right is of the form $u\le_\sigma u^m v$ where $v$ is a proper prefix of $u$. 
Notice that for $v=\varepsilon$ we are done. 

For the considered prefix $v$ of $u$ we may also have some inequality of the form $v\le_\sigma u^j w$ with $j\in \mathbb N$, in particular 
the proof of the inequality $u^k\le_\sigma u^\ell$ has such a consequence.
We analyze the inequalities of that form for all prefixes of $u$.      

On the set $P=\{ v \in A^* \mid v \text{ is prefix of } u\}$ we define the relation $\rightarrow$ as follows: $v\rightarrow w$ if there is $j \ge 0$ such that $v\le_\sigma u^j w$ and $|v|<|u^jw|$. 
(Notice that if $w=v$, then $j>0$.)
Since $\le_\sigma$ is a stable quasi-order
the relation $\rightarrow$ is transitive.   
As is discussed above, there is $v\in P$, $v\neq u$ 
such that $u\rightarrow v$ and so at least one of the following cases occurs. 

Case I: $u\rightarrow \varepsilon$. This means $u \le_\sigma u^j$ with $j>1$ and we are done.

Case II: there is $v\in P$ such that $u\rightarrow v\not = \varepsilon$, and there is no $w$ such that $v\rightarrow w$. 
In particular, $v\not =u$. Let $\bar{v}$ be the suffix of $u$ such that $v\bar{v}=u$. Since $v\not= \varepsilon$, we have $|\bar{v}| <|u|$ and by the induction assumption there is $p$ such that $\bar{v}\le_\sigma \bar{v}^p$. 
If we consider the proof of $u \le_\sigma u^j v$, then the consequence given by the prefix $v$ is
trivial equality (by the assumption that
there is no inequality of the form $v\le_\sigma u^j w$ with
$|v|< |u^j w|$). Then the consequence of the proof given by the suffix $\bar v$ is in the form
$\bar{v}\le_\sigma \bar{v} u^{j-1} v$. Now we use this inequality $(p-1)$-times to get  
$\bar{v}^{p-1} \le_\sigma (\bar{v} u^{j-1} v)^{p-1}=\bar{v} u^{(p-1)j-1} v$. Then we multiply 
the former inequality by $\bar{v}$ on the right and we get $\bar{v}^p \le_\sigma \bar{v} u^{(p-1)j}$. Since we assumed 
$\bar{v}\le_\sigma \bar{v}^p$, we also get $\bar{v} \le_\sigma \bar{v} u^{(p-1)j}$.
Finally, we multiply by $v$ on the left and obtain $u\le_\sigma u^{(p-1)j+1}$.

Case III: there is $v\in P$ such that $u\rightarrow v\not=\varepsilon$, and $v\rightarrow v$.
This means that there is $j>0$ such that $v\le_\sigma u^j v$. Now, it is enough to multiply the former inequality by the suffix $\bar{v}$ of $u$ on the right, and we get $u\le_\sigma u^{j+1}$.
\qed
\end{proof}

We try to show that an ordered semigroup 
$S$ is congenial whenever we have 
an onto homomorphism $\sigma : A^+ \rightarrow S$ determining \wqo\ $\le_\sigma$.
This means that for every homomorphism $\varphi :B^+ \rightarrow S$,
the set $L_\varphi=\{bwb \mid b\in B, w\in B^*, b \le_\varphi bwb  \}$ has to be unavoidable. 
Hence, every periodic infinite word 
$w^\infty$ must contain a factor from $L_\varphi$. In particular,
if $B$ contains $n$ letters $b_1, b_2, \dots, b_n$, then,
for the word $w=b_1b_2 \dots b_n$, there is an index $i\in\{1,\dots ,n\}$ 
and an integer $p\in\mathbb N$
such that $\varphi (b_i) \le \varphi (b_i (b_{i+1} \dots b_i)^p)$.  
Since the homomorphism $\sigma$ is onto, we may consider 
words $w_j\in A^+$ such that $\sigma(w_j)=\varphi(b_j)$. In this setting, 
we want to show that
$\sigma (w_i) \le \sigma (w_i (w_{i+1} \dots w_i)^p)$. 
In fact, we aim on the stronger inequality
$w_i \le_\sigma w_i (w_{i+1} \dots w_i)^p$. Proposition~\ref{p:quasiCR} is a special case of this property for $n=1$.
The following statement fulfills the sketched program.
 

\begin{theorem}
\label{th:ekvivalentni-podminky}
Let $S$ be a finite ordered semigroup.
Then the following conditions are equivalent:
\begin{itemize}
\item[(i)] There exists an alphabet $A$ and an onto homomorphism $\sigma :A^+ \rightarrow S$ such that $\le_\sigma$ is a well quasi-order.
\item[(ii)] There exists an alphabet $A$ and an onto  homomorphism $\sigma :A^+ \rightarrow S$ such that,
for every $n\in\mathbb N$ and $a_1, \dots , a_n\in A$, there exists 
$i\in\{1,\dots ,n\}$ and $p\in\mathbb N$ such that $a_i \le_\sigma a_i (a_{i+1}\dots a_n a_1\dots a_i)^p$.  
\item[(iii)] There exists an alphabet $A$ and an onto homomorphism $\sigma :A^+ \rightarrow S$ such that,
for every $n\in\mathbb N$ and $u_1, \dots , u_n \in A^+$, there exists $i\in\{1,\dots ,n\}$ and $p\in\mathbb N$ such that $u_i \le_\sigma u_i (u_{i+1}\dots u_n u_1\dots u_i)^p$.  
\item[(iv)] There exists an alphabet $A$ and an onto homomorphism $\sigma :A^+ \rightarrow S$ such that, for every $n\in\mathbb N$ and $u_1, \dots , u_n\in A^+$, there exists $i\in\{1,\dots ,n\}$ and $p\in\mathbb N$ such that $\sigma(u_i) \le \sigma \left( u_i (u_{i+1}\dots u_n u_1\dots u_i)^p\right)$.  
\item[(v)] For every $n\in\mathbb N$ and $s_1, \dots , s_n\in S$, there exists $i\in\{1,\dots ,n\}$
and $p\in\mathbb N$
such that $s_i \le s_i\cdot(s_{i+1}\cdots s_n\cdot s_1\cdots s_i)^p$. 
\item[(vi)] For every alphabet $B$ and a homomorphism $\varphi :B^+ \rightarrow S$ the relation 
$\le_\varphi$ is a well quasi-order.
\end{itemize}
\end{theorem} 
\begin{proof}
We show the implications from top to bottom.
The omitted implications are easy to see. 
In the conditions (i) -- (iv), the same pair $(A,\sigma)$ is employed. 

\uvoz{(i)$\implies$(ii)}: 
We consider a new alphabet $B=\{b_1,\dots , b_n\}$ of size $n$
and a homomorphism $\alpha : B^+ \rightarrow A^+$ such that
$\alpha(b_i)=a_i$ for all $i\in\{1,\dots ,n\}$. 
We denote the composition $\sigma \circ \alpha$ by $\varphi$. 
By Lemma~\ref{l:vetsi-abeceda}, we know that the relation
$\le_\varphi$ is a \wqo. In particular, the infinite word
$(b_1b_2\dots b_n)^\infty$ has a factor in 
$L_\varphi=\{bwb \,|\, b\in B, w\in B^*, b \le_\varphi bwb  \}$. Therefore, there is $i\in\{1,\dots ,n\}$ 
and $p\in \mathbb N$ such that $b_i\le_\varphi b_i (b_{i+1}\dots b_nb_1\dots b_i)^p$.
Finally, we get $\sigma(a_i)=\varphi(b_i) \le \varphi (b_i (b_{i+1}\dots b_nb_1\dots b_i)^p)=
\sigma(a_i(a_{i+1}\dots a_na_1\dots a_i)^p)$.

\uvoz{(ii)$\implies$(iii)}:
We apply the condition (ii) on the word $u=u_1u_2\dots u_n$ which we see as a concatenation of individual letters. So, there is $i\in\{1,\dots ,n\}$, $p\in\mathbb N$ and $u_i',u_i''\in A^*$
such that $u_i=u'_iau''_i$ and 
$a\le_\sigma a(u''_i u_{i+1} \dots u_n u_1 \dots u_{i-1} u'_i a)^p$.
If we multiply this inequality by the word $u'_i$ on left and by the word $u''_i$ on right, we get 
$u_i\le_\sigma u_i(u_{i+1} \dots u_n u_1 \dots u_{i-1} u_i)^p$.

\uvoz{(v)$\implies$(vi)}:
It follows from Corollary~\ref{cor:periodic-word}.
\qed
\end{proof} 

The condition (ii) from Theorem~\ref{th:ekvivalentni-podminky} was mentioned in~\cite{bucher}
in the setting of rewriting systems, namely it occurs as condition (c) in the concluding section.
Also, the condition in Proposition~\ref{p:quasiCR} is mentioned 
there as the condition (b). It is mentioned in~\cite{bucher} without proof that the conditions 
(a), (b) and (c) are equivalent. 
 
The equivalence of the conditions $(i)$ and $(vi)$ in Theorem~\ref{th:ekvivalentni-podminky}
gives the following result saying that whether the induced quasi-order $\lesigma$ is \wqo\ does not depend on the homomorphism $\sigma$ and it is indeed a property of the ordered semigroup.

\begin{corollary} 
\label{cor:nezavislost}
Let $\sigma : A^+ \rightarrow S$ be a homomorphism onto a finite ordered semigroup
$S$. Then $\le_\sigma$ is a \wqo\ if and only if the semigroup $S$ is congenial.
\qed
\end{corollary}

We get the following characterization of congeniality using the condition $(v)$ of Theorem~\ref{th:ekvivalentni-podminky}. 

\begin{corollary} 
\label{cor:pologrupove}
Let $S$ be an ordered semigroup.
Then $S$ is congenial if and only if for every $n\in\mathbb N$ and $s_1, \dots , s_n\in S$, 
there exists $i\in\{1,\dots ,n\}$
such that $s_i \le s_i \cdot (s_{i+1}\cdots s_n \cdot s_1\cdots s_i)^{\omega}$. \dukazjinde
\end{corollary} 

Unfortunately, it is not possible to bound $n$ in 
Corollary~\ref{cor:pologrupove}.
Indeed, there is a sequence of ordered semigroups $S_m$ such that $S_m$ satisfies the condition in Corollary~\ref{cor:pologrupove} if $n<m$ and does not satisfy the condition for $n=m$

\section{Conclusion}
\label{s:conclusion}

We have shown that for a homomorphism $\sigma\colon A^+\to S$ onto a finite ordered
semigroup, it is decidable whether $\le_\sigma$ is a \wqo. We also proved that the question
does not depend on $\sigma$, but it is indeed a property of the given ordered semigroup.
One may expect more effective or transparent characterization
similar to that of the unordered case in~\cite{kunc}. Nevertheless, our observations
suggest that such a characterization could be more difficult to obtain.

We conclude with a brief discussion of the applications of our results.
We do not see any direct impact of the research to the work done in~\cite{bucher}. 
On the other hand, in~\cite{kunc}, the \wqo\, was applied to prove regularity of maximal
solutions of very general language equations and inequalities (see also~\cite{pin-lata-2020}). 
The theory developed in \cite{kunc} may be naturaly extended to the ordered case, so our new class of ordered semigroups inducing well quasi-orders may find the application there.

\subsection*{Acknowledgement}

We are grateful to the referees for their numerous valuable suggestions 
which improved the paper, in particular, its introductory part. We also thank to Michal Kunc for inspiring discussions.

\newpage

\section*{Appendix}

\subsection*{The Proofs for Technical Lemmas from Section~\ref{s:preliminaries}}

\begin{lemmaagain}{l:vetsi-abeceda}
\label{a:l:vetsi-abeceda}
Let $\sigma\colon A^+ \rightarrow S$ be a homomorphism to an ordered semigroup $S$ such that 
$\le_\sigma$ is a \wqo. Let $B$ be an alphabet, 
$\alpha : B^+ \rightarrow A^+$ be a homomorphism of free semigroups such that 
$\alpha (B)\subseteq A$, and $\varphi=\sigma\circ \alpha$. 
Then the quasi-order $\le_\varphi$ is a \wqo.
\end{lemmaagain}
\begin{proof}
Let $(w_i)_{i\in \mathbb N}$ be a sequence of words over $B$. Then 
$(\alpha(w_i))_{i\in \mathbb N}$ is a sequence of words over $A$. Since $\le_\sigma$ is a \wqo,
there exist $k$ and $\ell$ such that $k<\ell$ and $\alpha(w_k)\le_\sigma \alpha(w_\ell)$. 
Its proof consists of a list of inequalities 
$a_i \le_\sigma u_i$ where $i\in\{1,\dots, m\}$.
The considered factorizations correspond to the factorizations
$w_k=b_1b_2 \dots b_m$ and $w_\ell=v_1v_2\dots v_m$ such that $b_i\in B$, $v_i\in B^+$,
$\alpha(b_i)=a_i$, $\alpha(v_i)=u_i$, because $\alpha$ maps letters to letters.
Now, $a_i \le_\sigma u_i$ means $\sigma(a_i)\le \sigma(u_i)$, which can be written 
as $\varphi(b_i)\le \varphi(v_i)$, that is $b_i \le_\varphi v_i$. 
Composing these inequalities for $i\in\{1,\dots, m\}$, we get $w_k\le_\varphi w_\ell$ and we are done.
\qed 
\end{proof}

\begin{lemmaagain}{l:evaluace}
\label{a:l:evaluace}
A semigroup $S$ is congenial if and only if
$\leq_{\eval_S}$ is a \wqo.
\end{lemmaagain}
\begin{proof}
The implication ``$\Longrightarrow$" is trivial. 

Let $\varphi : B^+ \rightarrow S$ be a homomorphism to a finite ordered 
semigroup $S$ such that $\leq_{\eval_S}$ is a \wqo.
We denote $C=\varphi(B)\subseteq S$ and consider the homomorphism 
$\alpha :C^+ \rightarrow S$ which is the restriction of $\eval_S$ to $C^+$. 
Then $\le_\alpha$ is a \wqo, since $\le_\alpha$ is a restriction of $\le_{\eval_S}$ on $C^+$.
Now, knowing that $\le_\alpha$ is a \wqo, we apply Lemma~\ref{l:vetsi-abeceda} 
and obtain that $\le_\varphi$ is \wqo\ too.\qed
\end{proof}

\subsection*{The Proof of Proposition~\ref{p:(1)}}

\begin{propositionagain}{p:(1)}
\label{a:p:(1)}
Every congenial semigroup $S$ satisfies the condition
\begin{equation}\tag{\ref{eqn:(1)}}
\label{a:eqn:(1)}
\forall s,t\in S\colon s\leq (s\cdot t)\naomeg\cdot s \text{ or } t\leq t\cdot(s\cdot t)\naomeg.
\end{equation}

\end{propositionagain}
\begin{proof}
Let $S$ be a congenial semigroup. Then $\leq_{\eval_S}$ is \wqo.
However, it is useful to distiguish a role of letters and elements of the semigroup $S$.
Thus we replace $S^+$ by an isomorphic semigroup $A^+$ and denote by $\sigma$ 
the homomorphism $\sigma\colon A^+\to S$.
This homomorphism has a useful property, namely, 
the restriction of $\sigma$ to $A$ is a bijection between $A$ and $S$. 
By the assumption, the relation $\le_\sigma$ is a \wqo.

We start the proof by showing that a congenial semigroup $S$ satisfies $s\leq s^{\omega +1}$ for every $s\in S$. 
Consider a sequence of words $\{(a^i)\}_{i=1}^{\infty}$, where $s=\sigma(a)$. Due to $\lesigma$ being a \wqo, there are some positive integers $k<\ell$ such that $a^{k}\lesigma a^{\ell}$. By definition of $\le_\sigma$ we have $a\lesigma a^{n+1}$ for some $n\in \bbN$. This means $s\leq s^{n+1}=s\cdot s^n\leq s^{2n+1}$ and by iteration we conclude that $s\leq s^{\omega+1}$.

We continue the proof with elements $s,t\in S$ and letters $a,b\in A$ such that $s=\sigma(a)$ and $t=\sigma(b)$.
Consider the sequence of words $\{(ab)^i\}_{i=1}^{\infty}$ over the alphabet $A$. 
Then for some positive integers $k<\ell$ it holds that $(ab)^k\lesigma (ab)^{\ell}$. 
Using the definition of the order $\lesigma$ we find 
the first non-trivial inequality of the form $ab\lesigma ab\dots $. 
This means that $a\le_\sigma au$ and $b\le_\sigma v$ for some $u,v\in A^*$ such that $auv$ is a prefix of $(ab)^\ell$ of the length at least $3$. Now, assume that $u\not=\varepsilon$.
Then from $a\le_\sigma au$ it follows $s\leq s\cdot t\cdots r$, where $r\in\{s,t\}$. 
In case of $r=s$, we have for some $p\in\bbN$ that $s\leq s\cdot (t\cdot s\cdots t\cdot s)=s\cdot (t\cdot s)^{p}$. Iterating this inequality, we get $s\leq s\cdot (t\cdot s)\naomeg$. 
In case of $r=t$, we have for some $p\in\bbN$ that $s\leq s\cdot t\cdot s\cdots s\cdot t=(s\cdot t)^{p}$. 
Taking the $\omega\mhyphen$power of both sides, we obtain $s\naomeg \leq (s\cdot t)\naomeg$. 
Now we multiply by $s$ on the right to get $s\naomeg\cdot  s\leq (s\cdot t)\naomeg \cdot s$. 
Using $s\leq s^{\omega+1}$ we conclude that $s\leq (s\cdot t)\naomeg \cdot s = s\cdot (t\cdot s)\naomeg$. 

Assume that $u=\varepsilon$. Then we have $b\le v$, where $v$ starts with $b$ and has the length at least $2$. Similar to the previous case we obtain $t\leq t\cdot s\cdots r$, and we may proceed in 
the same way.
\qed 
\end{proof}

\subsection*{The Additional Material for Section~\ref{s:conditions} }

Example~\ref{ex:brandt} gives a useful insight. Moreover, it inspires us to 
show that an arbitrary finite $0\mhyphen$simple semigroup can be ordered
in such a way that it is congenial. 
Recall that, by a $0\mhyphen$simple semigroup $S$ is a semigroup with 
zero element $0$ that has exactly two distinct ideals: $\{0\}$ and $S$. 
Then in a finite $0\mhyphen$simple semigroup 
we take $0$ as the top element that covers the antichain
of all non-zero elements. 
We denote the class of such ordered semigroups by $0\mhyphen\pv{CS}^-$.

\begin{proposition}\label{prop:nulajed}
Every finite ordered semigroup in $0\mhyphen\pv{CS}^-$ is congenial.
\end{proposition}

\begin{proof} 
We check that every finite $0\mhyphen$simple semigroup satisfies 
the inequality $ x \le x \cdot (y\cdot x )^\omega$, which gives the statement by
Proposition~\ref{prop:silnejis-nez-1}. 
Since $0$ is the top element, we have to discuss only the case 
when $x \cdot (y\cdot x )^\omega \not = 0$. However, in this case we have 
$ x \cdot (y\cdot x )^\omega = x $.
\qed 
\end{proof}

Another class of ordered semigroups to which the example $B_2^-$ fits very well is the class of inverse semigroups.
This class is extensively studied as it is a natural generalization of groups.
For reader not familiar with that theory, we just mention that 
every inverse semigroup may be represented as a semigroup of partial bijections
on an aprropriate set $X$
(the so-called Vagner-Preston representation theorem). 
In this representation, every element can be viewed as a  
relation, {\it i.e.}, the element is a subset of $X\times X$.    
Thus elements can be compared by the inclusion $\subseteq$ 
if they are represented by relations.
In this way, any inverse semigroup is implicitly ordered, 
however we worked with the dual order in Example~\ref{ex:brandt}. We call this dual order \emph{anti-natural}. 
Since every considered ordered 
inverse semigroup satisfies the inequality $x \le x\cdot (y \cdot x)^\omega$,
we obtain the next consequence of Proposition~\ref{prop:silnejis-nez-1}.

\begin{proposition}\label{prop:inverse}
Let $S$ be a finite inverse semigroup ordered anti-naturally. Then  $S$ is congenial.\qed
\end{proposition}

The following theorem describes a condition generalizing that from Proposition~\ref{prop:silnejis-nez-1}. The proof is based on the same ideas. 
After the proof we also explain that a semigroup ordered by the equality
satisfies the assumtions from the statement if and only if the semigroup 
is a chain of simple semigroups.

Before we formulate the result, we need to recall the definition of the one sided Green relation $\gr{L}$. For $s,t\in S$,
we have $s \gr{L} t$ if there are elements $x,y\in S^1$ 
such that $x\cdot s=t$ and $y\cdot t=s$. Under this definiton, the relation  
$\gr{L}$ is an equivalence relation. 
We also write $s \le_\gr{L} t$ if there is $y\in S^1$ 
such that $y\cdot t=s$. 
(The dual relations $\gr{R}$ and $\le_\gr{R}$ are not used in our paper.)

For a non-empty subset $X$ of a semigroup $S$ we write $\langle X\rangle $ for
a subsemigroup of $S$ generated by $X$. By $E(S)$ we mean the set of all idempotents of a semigroup $S$. 

\begin{theorem}\label{th:polstar}
Let $S$ be a finite ordered semigroup. 
If the semigroup $S$ satisfies 
\begin{equation}
\label{eqn:(polstar)}
(\forall F\subseteq S , F\neq \emptyset) (\exists s\in F) (\forall e\in E(\langle F \rangle)): e\le_{\gr{L}} s \implies s\leq s\cdot e,
\end{equation}
then $S$ is congenial. 
\end{theorem}

\begin{proof}
We deal with the situation when $A=S$ and $\eval_S$ is a considered homomorphism.
Let $v$ be an infinite word over
an alphabet $A$. Denote $v'$ its suffix which contains every letter infinitely many times. 
Now we take 
$$F=\{\sigma(a)\,|\,a\in A, a \text{ occurs in }v'\}$$
and 
the condition~(\ref{eqn:(polstar)}) gives us the choice of the letter $a$, which is contained in $v'$ infinitely many times. The rest of the proof is analogous to the 
proof of Proposition~\ref{prop:silnejis-nez-1}. \qed 
\end{proof}

We explain the situation when the considered semigroup is the chain of simple semigroups.
For every $F$ we take $s\in F$ in the lowest $\gr{J}\mhyphen$class. The idempotent $e$ satisfies $e\geq_{\gr{J}}s$ and with $e\le_{\gr{L}}s$ we have $e\gr{L}s$. The idempotent $e$ is a right neutral element in its $\gr{L}\mhyphen$class. We get $s\cdot e=s$ and we see that every chain of simple semigroups satisfies the condition~(\ref{eqn:(polstar)}).

On the other hand, if the semigroup $S$ is ordered by equality and satisfies 
the condition~(\ref{eqn:(polstar)}), then we may take $F=\{t,t'\}$.
Let $s\in F$ be the element satisfying
the condition~(\ref{eqn:(polstar)}). 
If $s=t$ then we take $e=(t'\cdot t)^\omega$, for which $t\le t\cdot (t'\cdot t)^\omega$ follows.
If $s=t'$ then we take $e=(t\cdot t')^\omega$, and obtain $t'\le t'\cdot (t\cdot t')^\omega$.
That means that the semigroup is a chain of simple semigroups.

\subsection*{The Additional Material for Section~\ref{s:automaton}}

\begin{theoremagain}{th:OKautomat}
\label{a:th:OKautomat}
Let $\sigma : A^+ \rightarrow S$ be a homomorphism onto a finite ordered semigroup
$S$.  Then $\le_\sigma$ is a \wqo\ if and only if the automaton 
$\mathcal{A}_{\sigma}$ does not contain 
an infinite path starting in the initial state $\iota$.
\end{theoremagain}

\begin{proof}
If $\le_\sigma$ is not a \wqo, then by Propositon~\ref{p:bucher} there exists an infinite word
avoiding $L_{\sigma}=\{ava \mid a\in A, v\in A^*, a\le_\sigma ava\}$. Thus, for its each finite prefix $w$, we know that 
$\delta(\iota,w)$ is defined in $\mathcal A_\sigma$ by Lemma~\ref{l:automat}. Hence there is 
an infinite path in $\mathcal A_\sigma$ starting in the initial state.

On the contrary, assume that $\mathcal{A}_{\sigma}$ contains an infinite path starting in the initial state $\iota$. Then the label of that path is an infinite word avoiding the set of words 
$L_\sigma$ by Lemma~\ref{l:automat}. Again, by Proposition~\ref{p:bucher}, we get that $\le_\sigma$ is not a \wqo.
\qed 
\end{proof}

We add an illustrative example of the construction of the automaton $\mathcal A_\sigma$.

\begin{example}
In Example \ref{ex:brandt} we saw that $B_2^+\notin \pekna$. 
If we consider $\sigma$ as a canonical mapping sending a letter 
from $A=\{a,b\}$ onto the element in the semigroup $B_2^+$ of the same name, 
then $S\times A\times A$ has 20 elements. 
The automaton $\mathcal{A}_{\sigma}$ consists of 
states $\iota, \bar a, \bar b$ and then states in $\mathcal P$. 
Every state $M\in\mathcal P$ is coherent, non-empty and disjoint with
$F=\{(a,a,a),(b,b,b)\}$. If we remove these two elements from 
$S\times A\times A$ and divide the resulting set into two parts by the last coordinate, 
we obtain two nine-elements sets 
$M_a=\{(s,x,a)\in S\times A\times A \mid s\not= a \vee x\not =a\}$, and 
$M_b=\{(s,x,b)\in S\times A\times A \mid s\not= b \vee x\not =b \}$.
Hence the number of states of $\mathcal{A}_\sigma$ is $3+2\cdot (2^9-1)=1025$. 

However, we are not interested in all states, but only those
reachable from the initial state $\iota$. We see that only tripples of the form 
$(\sigma(u),\first u, \last u)$, where $u\in A^+\setminus A$, may occur as elements of these states. 
Hence, in the previous considerations, we may replace $M_a$ by 
$\{(0,a,a),(0,b,a),(ba,b,a)\}$ and $M_b$ by 
$\{(0,a,b),(0,b,b),(ab,a,b)\}$. This gives a better uper bound for the size of the part of $\mathcal{A}_\sigma$ reachable from $\iota$. (The bound is $3+2\cdot (2^3-1)=17$.) 
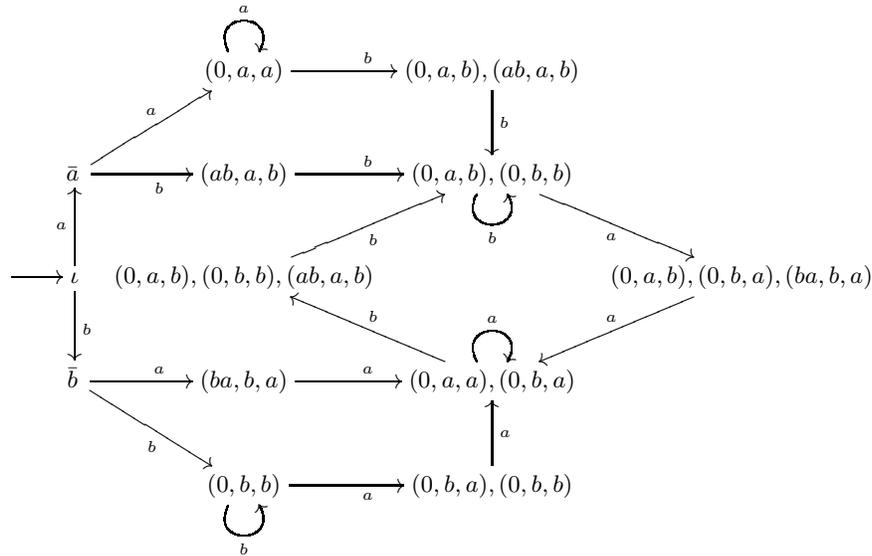
\begin{figure}[ht]
\centering
\[
\xymatrix@C=2mm{
&&& (0,a,a) \ar[r]^b  \ar@(ul,ur)[]^a & (0,a,b),(ab,a,b) \ar[d]^b & \\
&& \bar{a}\, \ar[ur]^a \ar[r]_b & (ab,a,b) \ar[r]^b & (0,a,b),(0,b,b) \ar[dr]_a \ar@(dl,dr)[]_b &\\ 
\ar[rr] && \iota  \ar[u]^a\ar[d]^b   & (0,a,b),(0,b,b),(ab,a,b) \ar[ur]_b & & (0,a,b),(0,b,a),(ba,b,a)  \ar[dl]_a \\
&&\bar{b}\,\ar[r]^a \ar[dr]_b & (ba,b,a)\ar[r]^a & (0,a,a),(0,b,a) \ar[ul]_b \ar@(ul,ur)[]^a & \\
&&& (0,b,b) \ar[r]_a \ar@(dl,dr)[]_b & (0,b,a),(0,b,b) \ar[u]_a &}
\]
    \caption{Reachable states in the automaton $\mathcal{A}_\sigma$ for the ordered semigroup $B_2^+$.}
    \label{fig:my_label2}
\end{figure}
Fig.~\ref{fig:my_label2} describes that part of the automaton $\mathcal{A}_\sigma$ completely. 
To simplify the notation, the states in $\mathcal P$
are labeled by the lists of elements (\textit{i.e.}, without curly brackets).

\end{example}

\subsection*{The proof of Corollary~\ref{cor:pologrupove}}

\begin{coragain}{cor:pologrupove} 
\label{a:cor:pologrupove}
Let $S$ be an ordered semigroup.
Then $S$ is congenial if and only if for every $n\in\mathbb N$ and $s_1, \dots , s_n\in S$, 
there exists $i\in\{1,\dots ,n\}$
such that $s_i \le s_i \cdot (s_{i+1}\cdots s_n\cdot  s_1\cdots s_i)^{\omega}$. 
\end{coragain} 
 \begin{proof}
With respect to Theorem~\ref{th:ekvivalentni-podminky}, it is enough to show that the two following conditions are equivalent for a given pair of elements $s,t\in S$:
$$(a)\ \  \text{there exists } p \in \mathbb N \text{ such that } s\le s \cdot (t\cdot s)^p;\qquad
(b)\ \ s\le s \cdot (t\cdot s)^\omega .$$
If (a) holds, then  we may iterate the inequality and obtain $$s\le s \cdot (t\cdot s)^p \le s \cdot (t\cdot s)^p \cdot (t\cdot s)^p \le 
\dots \le s \cdot (t\cdot s)^{pk}$$
for an arbitrary $k\in \mathbb N$. Since there is an integer 
$k\in \mathbb N$ such that $(t\cdot s)^k=(t\cdot s)^\omega$, we get 
the condition (b).  

The existence of such $k$ gives the implication \uvoz{(b)$\implies$(a)} directly. 
\qed
\end{proof}

\subsection*{The Sequence of Ordered Semigroups $S_m$}

For $m\ge 2$ we construct an ordered semigroup $S_m$ as follows. 
We describe the semigroup by the presentation 
over the $m$-letter alphabet $A_m=\{a_1,\dots a_m\}$: 
$$ S_m=\langle\, a_1, \dots , a_m \mid 
a_ia_j=0\  (\text{for } i,j \in\{1,\dots , m\} \text{ such that } 
j-i\not\in\{1, 1-m\}), $$
$$ \phantom{asaddddfsdfas}
(a_i\dots a_ma_1\dots a_{i-1})^2=a_i\dots a_ma_1\dots a_{i-1}
\ (\text{for } i\in\{1,\dots , m\} ) \rangle . $$
We explain the meaning of the previous semigroup presentation 
informally. We denote $u=a_1\dots a_m$.
The non-zero elements of $S_m$ are words of length at most $2m-1$ which are factors of $u^\infty$.
Every word which is not a factor of $u^\infty$ is
identified with $0$ by the rule $a_ia_j=0$ for appropriate indices. 
Now, any word which is a factor of $u^\infty$ of length at least $2m$
contains a factor $w$ of length $2m$ which is a conjugate of $uu$. Each such factor $w$ may 
be shortened  by the rule $(a_i\dots a_ma_1\dots a_{i-1})^2=a_i\dots a_ma_1\dots a_{i-1}$.
This procedure may be repeated until the resulting word has length smaller than $2m$.
The important property of the previous reduction is that the prefixes and suffixes of length $m$ are kept.
 
The description explains both the natural homomorphism $\sigma_m : A_m \rightarrow S_m$
and the multiplication on $S_m$. For the reader familiar with the structure theory 
of semigroups, we mention that $S_m$ has the minimal ideal $\{0\}$, one regular class and irregular singleton $\gr{J}$-classes consisting of factors of $uu$ of length smaller than $m$. The regular class consists of factors of $uuu$ of length at least $m$ and at most $2m-1$. It is a $m\times m$ box with singleton $\gr{H}$-classes and exactly $m$ idempotents (conjugates of $u$).

Finally, we introduce the order on $S_m$. At first, we put $0$ to be the largest element. Then
for two factors $x,y$ of $u^\infty$ of length smaller than $2m-1$, we have 
$$
x\le y\quad  \text{if} \quad  
\left\{ 
\begin{array}{l}
 x=y \text{  or}\\ 
 1<|x|< m \le |y| \ \wedge\  x \text{ is both the prefix and the suffix of } y .
 \end{array}\right. $$
One may check that the relation $\le$ is a stable order on the semigroup $S$. We point out that  
$a_i\le y$ only for $y\in \{a_i,0\}$.

In Corollary~\ref{cor:pologrupove}, we consider 
the following condition
for a given integer $n$:
\begin{equation}
\label{eq:pologrupaSm}
(\forall s_1, \dots , s_n\in S) (\exists i\in\{1,\dots ,n\}) : 
s_i \le s_i\cdot  (s_{i+1}\cdots s_n \cdot s_1\cdots s_i)^{\omega}. 
\end{equation}
If we want to test whether $S_m$ satisfies the condition~(\ref{eq:pologrupaSm}) for a given $n$, we distinguish two cases depending on whether $m$ divides $n$ or not. 

First assume that $m$ divides $n$. For every $i\in\{1,\dots, n\}$ we denote $[i,m]$ the remainder after dividing $i$ by $m$ and put $s_i=a_{[i,m]}$. Then $s_1\cdot s_2\cdots s_n=(a_1\cdots a_m)^{\frac{n}{m}}=a_1\dots a_m$ 
and we see that $S_m$ does not satisfy the condition~(\ref{eq:pologrupaSm}). 

If we assume that $m$ does not divide $n$, then 
$(s_{i+1}\cdots s_n\cdot  s_1\cdots s_i)^{\omega}\not=0$ only if $s_1 s_2\dots s_n$ is conjugate to some power
of $a_1\dots a_m$. This is not possible when each $s_i$ is an element in $A_m$. 
Thus there is an index $i$ such that $s_i\not\in A_m$, and the condition~(\ref{eq:pologrupaSm}) is valid. 

We conclude with the statement that $S_m$ satisfies  the condition~(\ref{eq:pologrupaSm}) if and only if $m$ does not divide $n$. In particular, 
$S_m$ satisfies the condition~(\ref{eq:pologrupaSm}) for $n<m$ and does not satisfy the condition for $n=m$.

\newpage
%
%
\nocite{*}
\bibliographystyle{splncs04}
\bibliography{KK-congeniality}

\end{document}